\documentclass[12pt]{article}
\usepackage[paper=letterpaper,margin=2cm]{geometry}
\usepackage{amsmath}
\usepackage{amssymb}
\usepackage{amsfonts}
\usepackage{amsthm}
\usepackage{newtxtext, newtxmath}
\usepackage{enumitem}
\usepackage{titling}
\usepackage[colorlinks=true]{hyperref}
\usepackage{verbatim}
\usepackage{xcolor}
\usepackage{natbib}
\usepackage{draftwatermark}
\SetWatermarkText{}
\SetWatermarkScale{4}
\usepackage{tikz}
\usepackage{pgfplots}
\usepackage{pgfplotstable}
\pgfplotsset{compat=1.7}
\usepackage{subcaption}
\usepgfplotslibrary{groupplots}

\usepackage{xcolor}
\DeclareUnicodeCharacter{2606}{\textcolor{red}{Error}}

\newtheorem{theorem}{Theorem}

\newtheorem{corollary}{Corollary}

\newtheorem{example}{Example}

\newtheorem{proposition}{Proposition}
\newtheorem{remark}{Remark}

\newcommand{\calf}{\mathcal{F}}

\begin{document}
\title{Haves and Have-Nots: A Theory of Economic Sufficientarianism }

\author{Christopher P. Chambers\thanks{Corresponding author.  Department of Economics, Georgetown University, ICC 580  37th and O Streets NW, Washington DC 20057. E-mail: \texttt{Christopher.Chambers@georgetown.edu}.} \and Siming Ye\thanks{Department of Economics, Georgetown University, ICC 580  37th and O Streets NW, Washington DC 20057. E-mail: \texttt{sy677@georgetown.edu}.}\thanks{This paper was originally circlated under the title ``On Sufficientarianism.''  We are grateful to the editor, Tilman Borgers, as well as an anonymous associate editor and three anonymous referees.  The current manuscript differs significantly from the original submission largely due to their comments.  They have done much to improve the manuscript. }}

\date{\today}

\maketitle
\begin{abstract}We introduce a generalization of the concept of sufficientarianism, intended to rank allocations involving multiple consumption goods.  In ranking allocations of goods for a fixed society of agents, sufficientarianism posits that allocations are compared according to the number of individuals whose consumption is deemed sufficient.  We base our analysis on a novel ethical concept, which we term sufficientarian judgment.  Sufficientarian judgment asserts that if in starting from an allocation in which all agents have identical consumption, a change in one agent's consumption hurts society, then there is no change in any other agent's consumption which could subsequently benefit society.  Sufficientarianism is shown to be equivalent to sufficientarian judgment, symmetry, and separability. We investigate our axioms in an abstract environment, and in specific economic environments.  Finally, we argue formally that sufficientarian judgment is closely related to the leximin principle.\end{abstract}

%We discuss three properties that uniquely pin down our family and, up to continuity requirements, endogenize the sufficientarian threshold.  Two, anonymity and reinforcement, are standard.  The third, sufficientarian monotonicity, generalizes classical monotonicity and requires that for any two lists of bundles in which individuals are ordered similarly by how much they receive, the worse of the two is indifferent to the new list of bundles considered by taking the componentwise minimum of the two bundles for each agent.\end{abstract}

\section{Introduction}
%Hook
%Can talk about a dev story for poverty trap，

%$\lhd$ $\unlhd$
Consider a society in which all individuals have access to enough clean water, food, and education.  Independently of whether any of them are wealthy or how the individuals obtained their positions, such an environment arguably dominates a society in which some of the individuals lack access to clean water and some lack access to food.  In this paper, we study a new form of economic \emph{sufficientarianism}, an interesting criterion of distributive justice recently introduced to the economics literature by \citet{Mariotti_2021} and \citet{Bossert2021, Bossert2022}. The goal of sufficientarianism is to ensure each individual has a sufficient level of well-being.  This work departs from the concept of well-being as the primitive, directly working instead with economic observables. 

%One could think of fixing a poverty level, and ensuring that as many individuals as possible reach that poverty line.  

The basic idea of sufficientarianism is that potential allocations of resources for a society are compared according to the number of individuals meeting an exogenously specified threshold of well-being.   The more agents who reach the level of well-being, the better.  %A sufficientarian ranking satisfies a weak version of the Pareto principle (if all agents weakly gain, then society weakly gains), but not the typical stronger version.  %Primitive to the sufficientarianism concept is a measure of well-being.  Different concepts of well-being lead to different concepts of sufficientarianism.  

Our work features several contributions.  First, we investigate a framework in which sufficientarianism is based on directly observable consumption, rather than well-being.  We envision observing consumption of a finite collection of individuals, we term this an allocation.  Our goal is to rank allocations.  We have in mind ranking allocations, for example, consisting of food or clean water.  Second, we provide a novel axiomatization of sufficientarianism in a totally abstract environment.  When describing the set of goods which are deemed ``sufficient,'' we allow total generality; the set can take any form whatsoever.  Importantly, this axiomatization endogenizes the sufficient set, rather than taking an exogenous sufficientarian threshold as a primitive in any of the axioms.  This axiomatization could also be understood as endogenizing the concept of welfare.  The set of ``sufficient'' consumption bundles consists of those whose welfare is deemed high enough, under the assumption that all agents have the same welfare for the same bundles.  Third, we also investigate the concept of sufficientarianism in concrete economic environments, so that if a consumption bundle $x$ is deemed sufficient, so is any bundle which is larger than $x$.  

Our work also describes an ethical condition that we believe to be novel and highly related to the concept of sufficientarianism:  for lack of a better term, we call this principle \emph{sufficientarian judgment}.  Roughly speaking, sufficientarian judgment asks that there should be an absolute priority on helping the worst off.  Here is how we formalize the idea.  We do not take individual preferences as a primitive.  Instead, our primitive is a ranking over allocations.  We need to be able to formalize the idea that one agent is worse off than another.  We choose to do so in a way that we feel is particularly weak:  suppose that we begin with an allocation in which all agents have identical consumption.  We can imagine a social planner as identifying all agents as equally well-off with such an allocation.\footnote{This seems to be the idea behind egalitarian-equivalence for example, see \citet{pazner}.}

Now, imagine changing one agent's consumption; call this agent Alice.  All remaining agents continue to consume the previously identical consumption.  If the allocation where Alice has the new consumption has rendered the society \emph{strictly worse off} than the original allocation, we can infer that Alice has been made worse off than the remaining agents.\footnote{Our framework envisions that this ``worse off'' or implicit welfare ranking is the same for all agents.}  Sufficientarian judgment asks that there is no way to improve upon this new allocation by altering any other agent's consumption.  We cannot compensate society either fully or partially for Alice's loss by rewarding Bob, say.  Hurting Bob would also not help society.  The only way to improve society's welfare is to alter Alice's consumption.

To this end, sufficientarian judgment is a very strong requirement, asking that society should prioritize the least fortunate in an absolute fashion.  But it asks more: rewarding more fortunate agents cannot benefit society in any sense.  Though this is a strong property, nevertheless, a particularly rich class of rankings satisfies it.  We describe some of these when discussing the independence of the axioms.
%{\color{teal}Simone can you add some of the references and concepts here?}

Our main result establishes that a ranking over allocations satisfies sufficientarian judgment, and two more familiar axioms (separability and symmetry) if and only if it is sufficientarian.

Previous works base their results on a principle called "prioritarian threshold," as well as other ethical axioms related to those of \citet{mariotti2009non, mariotti2012allocating, mariotti2013impossibility, alcantud2013, lombardi}. We will not go into detail about these axiomatizations, but one key point is that the prioritarian threshold specifically references a sufficient set, via a threshold. It does this in a context in which the observable is a collection of real-valued welfare levels, rather than allocations, and the axiom is parametric in a sense, taking as a parameter a threshold level that defines the sufficient set in this context.

In contrast, our work endogenizes the notion of a sufficient set. We can compare our work to the previous works by considering more structured economic environments. Though our main result takes place in an abstract environment, we develop several corollaries that apply to concrete economic environments.

For example, if consumption lies in a set with an ordered structure, then a natural monotonicity property implies that the sufficient set is closed upwards (upper-comprehensive). Our work extends naturally to classical commodity space, or more generally to any meet-semilattice. In this framework, we describe a natural generalization of the version of sufficietarianism presented in \citet{Mariotti_2021}. This generalization arises when we imagine that the sufficient set is determined by a threshold: There is a fixed threshold bundle $\beta$, and the sufficient set consists of all bundles at least as large as $\beta$. We can capture this version of sufficientarianism (called \emph{threshold sufficientarianism}) by a simple axiom, asking that different goods in commodity space cannot substitute for one another. That is, if a decrease in food for an agent renders society worse off, increasing her clean water will not bring society back to the starting point. Usually, this axiom, which we call $\geq$-complements, suffices to axiomatize threshold sufficientarianism, but sometimes an additional continuity condition is needed.

Finally, we also aim to show that sufficientarianism is related to the leximin property of ranking, via a weakening of sufficientarian judgment.  Imagine that we can assign a utility to each bundle and that the utility of a bundle can take a value of either $0$ or $1$.  Then maximizing any anonymous Bergson-Samuelson social welfare function will maximize the number of agents getting a utility of $1$.  This goes in particular for the leximin ranking and for the sufficientarian ranking.

However, utilizing a leximin ranking after assigning a utility to all bundles more generally results in a ranking that satisfies a weakening of sufficientarian judgment, which we term \emph{weak sufficientarian judgment}.  This property again starts with the idea that starting from an allocation in which all agents consume the same thing, a change in Alice's consumption hurts society.  Weak sufficientarian judgment allows that a subsequent change in Bob's consumption may cause a social gain, but can never bring society back to the initial social welfare level, where all agents were consuming the same thing.

Associating a utility with each bundle, and then ranking allocations according to the imputed leximin ranking between the allocations gives a ranking of allocations that satisfies this property.  We also show that, for two agents, such binary relations are the only ones satisfying weak sufficientarian judgment, symmetry, and separability.  We call these binary relations \emph{endogenous leximin}.

The paper proceeds as follows.  We conclude the introduction with some related literature; then we proceed directly to the model and abstract result (Section~\ref{sec:model}).  Section~\ref{sec:structured} investigates our concepts in concrete economic models, in particular proposing characterizations of threshold sufficientarianism.  Section~\ref{sec:leximin} describes the weakening of the sufficientarian judgment and the connection to the leximin ranking.  Finally, Section~\ref{sec:conclusion} concludes.

\subsection{Related Literature}

\citet{Mariotti_2021} introduced two substantive axioms, termed \emph{absolute individual improvement} and \emph{prioritarian threshold}.  We will not describe these axioms in detail.  The first of these asks that if one allocation $b$ strictly socially dominates another $a$, and in $b$ some agent's welfare is not maximal, then replacing that agent's welfare by the maximum possible welfare in $a$, and by anything else in $b$ does not lead to a reversal in the ranking.  The second of these, prioritarian threshold, is the one which references the sufficientarian threshold.  It asks for some level of welfare $\beta$ such that if any agent consumes strictly less than $\beta$, then replacing the welfare of all agents by $\beta$ makes society strictly better off.  Our axiom of sufficientarian judgment is logically independent of each of these two, and in a strong sense: in the setting of their paper, we can construct binary relations satisfying all of our axioms except separability, but neither of theirs, and we can construct binary relations satisfying all of their axioms except separability, but not sufficientarian judgment.

Aside from the previously mentioned work on sufficientarianism, \citet{Mariotti_2021, Bossert2021, Bossert2022}, the sufficientarian judgment axiom, as well as the notion of sufficientarianism, is closely related to a collection of lattice-theoretic axioms employed in characterizing similar representations.  The connection here is best understood considering a two-agent world.  Suppose that $a,b,c\in \mathbb{R}$, where these now represent welfare levels.  Sufficientarian judgment is the requirement that $(b,c)$ is indifferent to the socially worse of $(b,b)$ or $(a,c)$.  In a sense, the axiom requires that the binary relation commute with respect to ``meets,'' in a lattice-theoretic sense.  See for example \citet{miller2008group,chambers2014inefficiency,chambers2014scholarly,chambers2018}.  In these works, it is explained that the main idea is originally due to \citet{kreps1979} in his study of preference for flexibility over menus.  Related is the work of \citet{hougaard1998,christensen1999}, where the lattice-theoretic version of the property appears in the cardinal form.

The properties of sufficientarian judgment and weak sufficientarian judgment are similar in spirit to basic concepts such as the Pigou-Dalton transfer principle or other generalizations in economic environments, for example, \citet{fleurbaey2007two}, \citet{fleurbaey2014}, or \citet{sprumont}.  These pieces of literature present forms of absolute priority for the worst off.  To understand this claim, consider a two-agent economy.  Weak sufficientarian judgment asks that, starting from an allocation in which both agents consume the same, if we change $1$'s consumption, and get a worse social bundle, then any change in agent $2$'s consumption cannot bring society back to the welfare level of the starting point.  What this means is that the original bundle strictly dominates any bundle in which agent $1$'s imputed individual welfare has lowered from it.  

%In a sense, weak sufficientarian judgment is an extreme application of the
%The distinction here is that we do not work directly with classical commodity spaces, nor with welfare levels, but the ideas behind the principles are clearly related.  

What we call ``endogenous leximin'' is a concept that has existed in different forms in the literature for some time, although we are not aware of any work on this topic in a totally abstract environment.  Most (but not all) of this work supposes that individuals have different preferences, and represents these preferences via some type of canonical utility representation.  \citet{sprumont} describes a special case of endogenous leximin in an economic environment where individuals have preferences, and the common ranking underlying endogenous leximin is naturally in accordance with these preferences.  The monograph by \citet{fleurbaey} contains many references to works studying rankings of allocations where individuals' preferences differ, and where a generalization of endogenous leximin allows each individual to possess their own utility function.  For example, axiomatizations based on the classical money metric representation or on the representation first described by \citet{wold} have been studied.  In particular, they describe axiomatizations of this type of endogenous leximin ranking in several classical economic environments.  %Obviously, we can interpret our relation $\geq^*$ as a preference over alternatives that is common to all agents.

%His axiom of ``dominance aversion'' also bears resemblance to our notion of weak sufficientarian judgment.

%In extending sufficientarianism to a multi-dimensional framework, this paper navigates previously unexplored territory, diverging from the predominantly welfaristic focus in the literature. We recognize a parallel between our approach to multi-dimensional sufficientarianism and the broader economic themes of multidimensional poverty measures\footnote{See Atkinson, A. B. (2003). \textit{Multidimensional deprivation: contrasting social welfare and counting approaches}. The Journal of Economic Inequality, 1, 51-65; Alkire, S., and Foster, J. (2011). \textit{Counting and multidimensional poverty measurement}. Journal of public economics, 95(7-8), 476-487. These works illustrate various ways to measure poverty across different dimensions, providing methods for weighting and combining multiple indicators.} and Nussbaum's capability approach\footnote{Nussbaum, Martha (2011). \textit{Creating capabilities: the human development approach}. Cambridge, Massachusetts: The Belknap Press of Harvard University Press. Nussbaum's capability approach emphasizes the importance of considering a variety of human abilities or "capabilities" in evaluating well-being, aligning with our multi-dimensional interpretation of sufficientarianism.}, which also emphasize the complexity of well-being across multiple dimensions. Our approach builds on these foundations by considering sufficientarian orderings in a multi-dimensional setting, capturing the nuanced relationships between dimensions.

In extending sufficientarianism to a multidimensional framework, this paper navigates previously unexplored territory, diverging from the predominantly welfaristic focus in the literature. We recognize a parallel between our approach to multidimensional sufficientarianism and the broader economic themes of multidimensional poverty measures \citep{atkinson2003, alkire2011} and Nussbaum's capability approach \citep{nussbaum2011}, which also emphasize the complexity of well-being across multiple dimensions. Our approach builds on these foundations by considering sufficientarian orderings in a multidimensional setting, capturing the nuanced relationships between dimensions.

\section{The Model and Results}\label{sec:model}

In this section, we explore the concept of economic sufficientarianism in an abstract environment.

Let $A$ be a nonempty set, denoting the set of possible consumption bundles. For example, $A$ may consist of a classical commodity space, it may consist of unordered bundles of goods, or it may consist of a set of life outcomes. Alternatively, $A$ may index characteristics that may be associated with a given individual:  health status, life satisfaction, and so forth.  Finally, in accordance with the classical study of sufficientarianism, $A$ may consist of welfare levels.  Elements of $A$ are the objects which are assigned to agents.  Typical elements of $A$ will be denoted by $a,b\in A$.

Let $N=\{1,\ldots,n\}$ be a finite set, where $n\geq 2$ denotes the number of individuals in the society.  $N$ represents the set of agents in society.  

We can think of elements of $A^N$ as \emph{ allocations}.  Every $x\in A^N$ associates a unique member of $A$ with each member of society.  

For $x\in A^N$, we typically refer to the $i$th element of $x$ as $x_i\in A$.  Our goal is to rank members of $A^N$.  To this end, given is a binary relation $\succeq$ on $A^N$.  This binary relation is our main objective of study.  We now turn to the properties we ask of this relation.

The following two are standard properties, though the second of them is not explicitly stated in \citet{Mariotti_2021}.  It follows as a consequence of some of their other axioms.

\textbf{Weak order}: $\succeq$ is a weak order.\footnote{That is, $\succeq$ is complete and transitive.}

\textbf{Symmetry}:  Let $\sigma:N\rightarrow N$ be a bijection and let $a\in A^N$.  Then $a \sim a \circ \sigma$.

As in \citet{Mariotti_2021}, we postulate a simple separability condition on $\succeq$.  The interpretation of this axiom is much the same as it is in that paper.

%Following Alcantdu et al.... 

To this end, for any $x,y\in A^N$ and $M\subseteq N$, $(x_M,y_{-M})$ denotes the member of $A^N$ for which $(x_M,y_{-M})_i = x_i$ when $i\in M$ and $(x_M,y_{-M})_i=y_i$ when $i\notin M$.

%For any $i\in N$, $x_i \in X$ and $z\in X^N$, we write $(x_i,z_{-i})$ to denote the member of $X^N$ for which $(x_i,z_{-i})_i = x_i$ and if $j\neq i$, then $(x_i,z_{-i})_j = z_j$.  Similarly, for 

%Our version of separability is written in a fashion that appears to be apparently stronger than the one in \citet{Mariotti_2021}, a standard induction argument establishes that it is equivalent to theirs under the weak order hypothesis.

\textbf{Separability}: For any $M\subseteq N$, and any $x,y,x',y'\in A^N$, $(x_M,y_{-M})\succeq (x'_M,y_{-M})$ \textit{iff} $(x_M,y'_{-M})\succeq (x'_M,y'_{-M})$.

This form of separability is written in a fashion that is apparently stronger than the one in \citet{Mariotti_2021}. However, a standard induction argument establishes that it is equivalent to theirs under the weak-order hypothesis.

We now introduce our main innovation. First, some notation. Let $a,b,c\in A$, and $i,j\in N$, for which $i\neq j$.  Then $b$, with an abuse of notation, also represents the member of $A^N$ for which $b_k = b$ for all $k\in N$.
We use the notation $a_i b$ to represent that member of $A^N$ for which $(a_i b)_i = a$ and for all $k\neq i$, $(a_i b)_k = b$, and $c_j(a_i b)$ to be the member of $A^N$ for which $(c_j (a_i b))_i = a$, $(c_j(a_i b))_j= c$ and for all $k\notin N\setminus \{i,j\}$, $(c_j (a_i b))_k = b$.

%To do so, we first need to define a notion of what it means for one bundle $a\in A$ to be better than another bundle $b\in B$.  We prefer to do so without referencing separability.  

%For the next property, for any $a,b\in A$, define $a\succeq^* b$ if $(a,\ldots,a)\succeq (b,\ldots,b)$.  Similarly, for $x\in X$, $a\in A$, and $j\in N$, define $a_jx$ to be the member of $X$ for which $(a_jx)_i=a$ when $i=j$ and $(a_jx)_k=x_k$ otherwise.

\textbf{Sufficientarian judgment:} Let $a,b,c\in A$ and suppose that $i\neq j$.  Then $b \succ a_i b$ implies $a_i b \succeq c_j (a_i b)$.  

%Let $a,b,c\in A$ and $i,j\in N$, with $i\neq j$.  Let $x=(b,\ldots,b)$.  If $x \succ a_ix$, then $a_ix \succeq c_j(a_ix)$.

In words, sufficientarian judgment says the following.  Starting from an allocation in which all agents are treated identically in terms of consumption, receiving $b$, suppose a change in $i$'s alternative to $a$ forces social welfare to drop strictly.  Implicitly, this must happen because $i$ becomes worse off than the other agents.  What sufficientarian judgment asks is that this drop in social welfare can never be compensated, either fully or partially, by giving some other agent $j$ consumption of $c$.  This is true no matter how much of a gain $c$ would cause to agent $j$.  Society should not treat $i$ and $j$'s welfare as substitutes, at least in the case in which $i$'s welfare level has become less than $j$'s.  

Sufficientarian judgment applies only when the starting position of all agents is identical (the alternative $b\in A$ is consumed by all agents), though the intuition would also be valid in case the starting position of $i$ is worse than that of agent $j$.  That said, we can only meaningfully speak of a starting position of $i$ being worse than that of $j$ under the hypotheses of symmetry and separability. % We wish to refrain from relying on these axioms in formulating the condition.

In the context of sufficientarian judgment, we can infer that agent $i$ is worse off than the remaining agents precisely because in moving from $b$ to $a_ib$, society has become worse off.  This involves an implicit assumption that all members of society were equally well off when consuming $b$; an implicit interpersonal comparison.  The idea behind the axiom then is that society should prioritize raising the welfare of poor agents as much as possible, to the extent that any gain in welfare for a richer agent would have no effect on the social ranking.  

On the other hand, if a loss in social welfare results from a loss of welfare of a rich agent, we want to allow society to benefit from a gain in welfare to a poor agent.  To understand, suppose $A = \mathbb{R}$, where $N=\{1,2\}$, and that we are considering $x\in A^N$ given by $(5,0)$.  It is reasonable to assume that $(5,0) \succ (0,0)$; so that a loss of $5$ units for the richer agent (agent $1$) hurts society.  At the same time, it is reasonable to suppose that $(0,5) \succ (0,0)$ as well: that loss can be compensated for by a gain in utility by agent $2$.

What sufficientarian judgment rules out is a situation whereby, for example, $(5,5) \succ (0,5)$ and $(0,7)\succ (0,5)$.  In this example, agent $1$'s move from $5$ to $0$ has decreased social welfare.  Starting from an implicit assumption that the two agents are equally well off under $(5,5)$, this means that agent $1$ has become worse off in the move to $(0,5)$.  She is now the ``poorer'' of the two agents.  We want to preclude that this drop in wealth for agent $1$ can be compensated for by enriching agent $2$, which rules out $(0,7) \succ (0,5)$.

Thus, it asserts that the priority should be on bringing the agent whose welfare dropped back up to the welfare of the remaining agents; in this sense it is a strongly egalitarian property.

We say that a binary relation $\succeq$ is \emph{sufficientarian} if there exists a set $S\subseteq A$ such that for all $x,y\in A^N$, $x \succeq y$ iff $|\{i\in N:x_i \in S\}|\geq |\{i\in N:y_i \in S\}|$.  We refer to the set $S$ as the \emph{sufficient set}.

This formalizes the notion of economic sufficientarianism.  The set $S$ becomes a parameter of the binary relation, the goal being to ensure as many agents as possible have consumption in $S$.

The sufficient set $S$ may be empty (or the entire set), in which case the binary relation $\succeq$ coincides with total indifference.  In case this is a concern, the following (standard) property rules this out.

\textbf{Non-degeneracy}: There exists $x,y\in A^N$ for which $x\succ y$.

\begin{theorem}\label{thm:mainresult}A binary relation $\succeq$ is sufficientarian \textit{iff} it satisfies weak order, symmetry, separability, and sufficientarian judgment.  These axioms are independent.\end{theorem}

\begin{proof}Let $\succeq$ satisfy the axioms.  Define the ranking $\geq^*$ on $A$ by $a \geq^* b$ iff $a_i c \succeq b_i c$ for some $c\in A$.  Observe that by separability, $a_i c\succeq b_i c$ iff $a_i d\succeq b_i d$ for any $c,d\in A$.  Obviously $\geq^*$ is a weak order.

We claim that $\geq^*$ has at most two indifference classes.  To this end, and by contradiction, suppose that there are $a^*,b^*,c^*\in X$ for which $a^*>^* b^* >^* c^*$.  

We invoke separability.  First we simplify notation.  Fix agents $1,2\in N$.  We define the binary relation $\succeq^2$ on $A^{\{1,2\}}$ by $(a,b)\succeq^2 (c,d)$ if for any $x,y\in A^N$ for which $x_{\{1,2\}}=(a,b)$ and $y_{\{1,2\}}=(c,d)$, where $x_i = y_i$ for all $i\in N\setminus \{1,2\}$, we have $x \succeq y$.  By separability and weak order, $\succeq^{\{1,2\}}$ is itself a weak order.

Now observe that $(b^*,b^*) \succ^2 (b^*,c^*)$ by separability (since $b^* >^* c^*$).  So, $b^* \succ c^*_2b^*$, also by separability.  By sufficientarian judgment, $c^*_2b^* \succeq a^*_1(c^*_2b^*)$, and so by separability it follows that $(b^*,c^*) \succeq^2 (a^*,c^*)$.  But this contradicts separability (as $a^* >^* b^*$).

%Recalling the definition of $\succeq^2$, and invoking sufficientarian judgement, we claim that the first case implies $(b^*,c^*)\sim^2 (b^*,b^*)$.  

%But $y^*\succ^* z^*$, so by Lemma~\ref{lem:separability}, and the definition of $\succeq^2$, we have $(y^*,y^*)\succ^2 (x^*,z^*)$, a contradiction.  The second case implies $(y^*,z^*)\sim^2 (x^*,z^*)$.  Again by Lemma~\ref{lem:separability} and the definition of $\succeq^2$, we have $(x^*,z^*)\succ^2 (y^*,z^*)$, also a contradiction.

So there can be at most two indifference classes for $\geq^*$.  Denote by $S\subseteq A$ the higher of these two indifference classes.  

Let $x,y\in A^N$.  We want to claim that if $|\{i\in N:x_i \in S\}|\geq |\{i\in N:y_i\in S\}$, then $x \succeq y$, similarly if $|\{i\in N:x_i \in S\}|> |\{i\in N:y_i \in S\}|$, then $x \succ y$.  By symmetry and transitivity, we may assume without loss that for each of $x$ and $y$, we have $x_1 \geq^* x_2 \ldots \geq^* x_n$ and $y_i \geq^* y_2 \ldots \geq^* y_n$.  

We argue by induction.  Suppose that $|\{i\in N:x_i \in S\}|\geq |\{i\in N:y_i \in S\}|$.  This implies that for all $i\in N$, we have $x_i \geq^* y_i$.  Consequently, by repeatedly invoking separability, $(x_1,x_2,\ldots,x_n) \succeq (y_1,x_2,\ldots,x_n) \succeq (y_1,y_2,x_3,\ldots,x_n)\ldots \succeq (y_1,\ldots,y_n)$ so that $x \succeq y$ by transitivity.  Similarly, if $|\{i\in N:x_i \in S\}|> |\{i\in N:y_i \in S\}|$, then there is in addition some $i\in N$ for which $x_i >^* y_i$.  Consequently, in the preceding chain, one of the instances of $\succeq$ is actually $\succ$ and so $(x_1,\ldots,x_n) \succ (y_1,\ldots,y_n)$ by weak order.

For the converse, weak order, symmetry, and separability are easily verified.  %Consider now $x,y\in X^N$, defined as in the hypothesis of sufficientarian judgment (so that for all $i,j\in N$, $x_i \succ^* x_j$ implies $y_i \geq^* y_j$).  Without loss, we may assume that $x_1 \geq^* x_2 \ldots \geq^* x_n$.  It follows that $y_1 \geq^* y_2 \ldots \geq^* y_n$.  

%Clearly $\geq^*$ has exactly two indifference classes, the higher of which corresponds to $S$.  We claim that $\{i\in S:y_i \in S\}\subseteq \{i\in N:x_in \in S\}$.  This follows as the set $\{i\in N:x_i\in S\}$ and $\{i\in N:y_i\in S\}$ are each (possibly empty) intervals in $S$, which whenever nonempty contain $1$.  The fact that $|\{i\in N:x_i\in S\}|\geq |\{i\in N:y_i \in S\}|$ implies that $\{i\in N:y_i \in S\}\subseteq \{i\in N:x_i \in S\}$.  Consequently, in defining $z$ as in the hypothesis of sufficientarian monotonicity, we have $\{i\in N:z_i \in S\}=\{i\in N:y_i\in S\}$, so that $z\sim y$.   

For sufficientarian judgment, let $a,b,c\in A$. Suppose that $b \succ a_ib$.  By definition, this means that $b\in S$ and $a\notin S$.  Since $b\in S$, in either case, whether $c\in S$ or $c\notin S$, we have $a_i b\succeq c_j(a_ib)$.  

On the independence of the axioms:  We illustrate by examples that the three axioms are independent. These examples will illustrate the general structure of rules satisfying these axioms.

\textbf{Weak order}\\
Let $\succeq$ be the relation defined by $x \succeq y$ \textit{iff} $y = x \circ \sigma$ for some permutation $\sigma:N\rightarrow N$.  Symmetry and sufficientarian judgment are automatically satisfied (sufficientarian judgment because the antecedent condition is never satisfied, since there are no strict rankings according to this relation).  Separability is similarly easy to verify.

\textbf{Symmetry}\\
Say that a binary relation $\succeq$ is \emph{weighted sufficientarian} if there is a function $\lambda:N\rightarrow\mathbb{R}_{++}$ and a set $S \subseteq A$ such that for all $x,y\in A^N$, $x\succeq y$ if and only if $\sum_{i\in N}\lambda_i \mathbf{1}_{\{x_i\in S\}} \geq \sum_{i\in N}\lambda_i \mathbf{1}_{\{y_i\in S\}}$.\footnote{Or just an ordinal ``qualitative measure'' ranking the sets.}  A related generalization, which would feature each individual possessing their own filter $\calf_i$, generally fails sufficientarian monotonicity.

Another class of solutions are \emph{dictatorships}, where there is a weak order $\geq$ on $A$ such that for all $x,y\in A^N$, $x \succeq y$ iff $x_i \geq y_i$.%a relation $\succeq_i$ for which $\geq \subseteq \succeq_i$, so that for any $x,y\in A^N$, $x \succeq y$ if and only if $x_i \succeq_i y_i$. 

%Still other, hybrid, variations of these two methods would satisfy the remaining axioms.  We leave this to future research.

\textbf{Separability}\\
Binary relations violating reinforcement are similarly many.  A canonical example would be obtained in the case of $A = [0,1]$, whereby for all $x,y\in A^N$, $x \succeq y$ if and only if $\min_i x_i \geq \min_i y_i$.  %For a more general set $X$, one would construct a $\leq$-chain $\mathcal{C}$ with the property that for every $x\in X$, there is a uniquely $\leq$-maximal element of $\mathcal{C}$.  Then define, for any $N\in\caln$ and $x,y\in X^N$, $x\succeq_N y$ if and only if said maximal element for $\bigwedge_{i\in N}x_i$ $\leq$-dominates the maximal element for $\bigwedge_{i\in N}y_i$. However, even in the case of $X = [0,1]$, many other rules are easy to describe. For any $N\in\caln$ and $x\in X^N$, and any $n\in \{1,\ldots,|N|\}$, let $x^*(n)$ denote the $n$-th highest value of $x$.  Then define $U_N:X^N\rightarrow\mathbb{R}$ as $U_N(x) = \inf_n nx^*(n)$.  For any $x,y\in X^N$, assign $x\succeq_N y$ if and only if $U_N(x) \geq U_N(y)$.  We leave the study of this class of rules to future research.\\

\textbf{Sufficientarian judgment}\\
There are many binary relations that violate the sufficientarian judgment.  Let $u:A\rightarrow \mathbb{R}$ be a function and define $x\succeq y$ iff $\sum_i u(x_i) \geq \sum_i u(y_i)$.  Then as long as the image of $u$ contains at least three values, this binary relation violates sufficientarian judgment.

\end{proof}

\begin{remark}Sufficientarian judgment could be replaced by the following dual condition in Theorem~\ref{thm:mainresult} or any of the results to follow.  The proof is essentially identical to that of Theorem~\ref{thm:mainresult} and is thus omitted.

\textbf{Dual Sufficientarian Judgment}: Let $a,b,c\in A$ and suppose that $i\neq j$.  Then $a_i b \succ b$ implies $c_j (a_i b) \succeq a_i b$.%$a_i b \succeq c_j (a_i b)$.  

The interpretation of this dual condition has less ethical content than does sufficientarian judgment but asks that if social welfare goes up in replacing agent $i$'s consumption with $a$ from $b$, then one cannot bring it back down by changing some other agent $j$'s consumption.  
\end{remark}

\section{Structured economic environments}\label{sec:structured}

In this section, we suppose that $A$ represents a more classical economic environment, for example, a commodity space.  To this end, we suppose it is endowed with a binary relation $\geq$, which we will assume to be a preorder.\footnote{That is, reflexive and transitive.  These assumptions are without loss, if $\geq$ is not a preorder the following monotonicity assumption would be equally valid for the smallest preorder containing it.}

The framework of \cite{Mariotti_2021} conforms to this model:  their assumption is that $A=[0,1]$, and they explicitly reference axioms relating to $\geq$.  Chief among these is the following weak form of monotonicity, asking that increases in consumption for all agents never hurt society.  

\textbf{Monotonicity}: Let $x,y\in A^N$ for which for all $i\in N$, $x_i \geq y_i$.  Then $x \succeq y$.

Say that a subset $S\subseteq A$ is \emph{monotone} (or \emph{upper comprehensive}) if $a\in S$ and $b\geq a$ implies $b\in S$.

We say that a ranking is \emph{monotone-sufficientarian} if there is a monotone set $S\subseteq A$ such that for all $a,b\in A^N$, $a \succeq b$ iff $|\{i\in N:a_i \in S\}|\geq |\{i\in N:b_i \in S\}|$.\footnote{This family of rankings was suggested to us by the anonymous associate editor.}

The following Corollary of Theorem~\ref{thm:mainresult} asserts that the addition of monotonicity to the preceding axioms ensures that the sufficient set $S$ is monotone.  Its proof is obvious.

\begin{corollary}A binary relation $\succeq$ is monotone sufficientarian \textit{iff} it satisfies weak order, symmetry, separability, sufficientarian judgment, and monotonicity.
\end{corollary}

We are now in a position to compare our results with those of \citet{Mariotti_2021}.  Say a binary relation $\succeq$ is \emph{threshold sufficientarian} if there exists $b\in A$ for which for all $x,y\in X^N$, $x \succeq y$ iff $|\{i\in N:x_i \geq b\}|\geq |\{i\in N:y_i \geq b\}|$.

Let us now consider the specific environment in which $A = [0,1]$ with its usual order $\geq$; this is the environment studied by \citet{Mariotti_2021}.  Every threshold sufficientarian binary relation is clearly monotone sufficientarian.  However, there are monotone sufficientarian binary relations which are not threshold sufficientarian.  For example, fixing some $b\in A$, we may consider a binary relation for which for all $x,y\in A^N$, $x \succeq y$ iff $|\{i\in N:x_i > b\}|\geq |\{i\in N:y_i > b\}|$.  A rule like this might be called ``strictly sufficientarian,'' but the main distinction is that only the individuals getting strictly more than $b$ have a sufficient amount.  Those getting exactly $b$ do not consume within the sufficient set.  

In the next section, we will discuss conditions that rule out these more general binary relations. However, the key point is that the distinction between the monotone sufficientarian rules and threshold sufficientarian rules in a single-dimensional environment like this one is largely technical.

To this end, we will now discuss the axioms of \citet{Mariotti_2021}.  Our goal is to show that there is no logical relation between their axioms and sufficientarian judgment.  We will do so by establishing that there are binary relations satisfying all of our axioms, except separability, and which violate both of their axioms.  We will also show that there are binary relations satisfying all of their axioms, except separability, and which violate sufficientarian judgment.  

As a first point, they also invoke the weak order, monotonicity, and separability axioms.  The axioms of theirs which differ from ours are as follows.  The motivation of these properties can be found in \citet{Mariotti_2021}.

\textbf{Absolute individual improvement}: Let $x,y\in A^N$ and let $i\in N$.  Suppose that $x \succ y$ and $x_i < 1$.  Then for all $b_i\in A$, $(1,x_{-i}) \succeq (b_i,y_{-i})$.

\textbf{Prioritarian threshold}:  There exists $\beta\in (0,1)$ for which for all $x\in A^N$ for which there exists $i\in N$ such that $x_i < \beta$ and for all $j,k\in N\setminus\{i\}$, $x_j = x_k$, we have $\beta \succ x$.

Our first example is a class of binary relations satisfying all of our axioms except for separability, but which satisfies neither of the \cite{Mariotti_2021} axioms.  This class also satisfies homotheticity whenever homothetic expansions are meaningful.  We leave the verification of these claims to the reader.  

Roughly, this class is a kind of generalized maxmin class, whereby any two allocations which are comonotonic (that is, if agent $i$ is strictly richer than $j$ in one allocation, $i$ is weakly richer than $j$ in the other) are judged according to a weighted maxmin criterion.  The binary relation is then completed by symmetry.

\begin{example}Let $\alpha = (\alpha_1, \alpha_2, ..., \alpha_N)$ be a fixed increasing sequence of real numbers. For any $x \in [0, 1]^N$, let $d(x)$ denote the non-increasing rearrangement of $x$. Define a utility function $U:A^N\rightarrow \mathbb{R}$ via $U(x) = \min_{i\in N} \{ \alpha_i d(x)_i \}$.  Finally let $\succeq$ be the binary relation on $A^N$ represented by $U$.\end{example}

\begin{figure}[ht]
  \centering
  \begin{tikzpicture}

  % Axes
  \draw[->] (0,0) -- (8,0) node[right] {Agent 1};
  \draw[->] (0,0) -- (0,8) node[above] {Agent 2};

  % 45 degree line
  \draw[dashed] (0,0) -- (8,8);

  % 30 degree line
  \draw[dashed, domain=0:8] plot (\x, {\x*tan(30)});
  
  % 60 degree line
  \draw[dashed, domain=0:8/tan(60)] plot (\x, {\x*tan(60)});

  % Points and lines
  \foreach \radius in {4, 2, 6} {
    \pgfmathsetmacro{\xthirty}{\radius*cos(30)}
    \pgfmathsetmacro{\ythirty}{\radius*sin(30)}
    \pgfmathsetmacro{\xsixty}{\radius*cos(60)}
    \pgfmathsetmacro{\ysixty}{\radius*sin(60)}

    \draw[thick] (\xthirty, \ythirty) -- (\xthirty, \ysixty) -- (\xsixty, \ysixty);
    \draw[thick] (\xthirty, \ythirty) -- (8, \ythirty);
    \draw[thick] (\xsixty, \ysixty) -- (\xsixty, 8);

    \fill (\xthirty, \ythirty) circle[radius=2pt];
    \fill (\xsixty, \ysixty) circle[radius=2pt];
    \fill (\xthirty, \ysixty) circle[radius=3pt];
  }

\end{tikzpicture}
  \caption{Level sets for Example 1}
  \label{fig:example1}
\end{figure}

Our second example is a binary relation satisfying all the axioms of \cite{Mariotti_2021} except for separability, but which does not satisfy the sufficientarian judgment.  The example also satisfies symmetry.  The example is easy to generalize; roughly, it coincides with threshold sufficientarianism except when for both allocations, all agents reach the sufficientarian threshold.  In this case, it ranks things monotonically but in a way in which all allocations for which at least one agent consumes $1$ are indifferent. Again, we leave the verification of these claims to the reader.

\begin{example}First, fix $\tau \in (0,1)$. The binary relation $\succeq$ behaves as the sufficientarian rule with respect to $\tau$ for any $x, y \in A^N$ where $\min_{i\in N}\min\{x_i, y_i\} < \tau$.

If for all $i \in N$, $\min\{x_i, y_i\} \geq \tau$, then we define $x \succeq y$ if and only if $\prod_{i \in N} (1-x_i) \leq \prod_{i \in N} (1-y_i)$.\end{example}

\subsection{On threshold sufficientarianism}

In this section, we propose some simple axioms that allow us to meaningfully talk about threshold sufficientarianism in a multidimensional environment.  To this end, we assume that $A$ is endowed with a partial order $\geq$.\footnote{A partial order is complete, transitive, and antisymmetric.}  Further, we assume that for each pair $a,b\in A$, there is a unique greatest lower bound (or \emph{meet}) according to $\geq$; we denote this greatest lower bound by $a \wedge b$.  Formally, we assert that the pair $(A,\geq)$ is a meet-semilattice.  Definitions here are standard, see \emph{e.g.} \citet{davey_priestley_2002}.

Now, we can formulate a simple axiom that requires that different goods cannot substitute for each other in the sufficient set.  Obviously, whether this axiom makes any sense depends on what we imagine consumption goods to be.

To understand the property, let us consider the standard meet-semilattice of $\mathbb{R}_+^2$.  

Our goal is to require the following property, stated informally: If the loss of one good for an agent makes society strictly worse off, then we cannot revert to the original level of welfare by compensating the same agent with any other good.

To see how this works in the $\mathbb{R}_+^2$ case, consider $(a_1,a_2)$, $(b_1,a_2)$ and $(b_1,b_2)$, where $a_1 > b_1$ and $a_2 < b_2$.  Suppose now that society determines that changing $i$'s consumption from $(a_1,a_2)$ to $(b_1,a_2)$ leads to a strict decrease in social welfare.  That is, the loss of good $1$ for agent $i$ renders the society worse off.  Then we want to claim that there is no amount of good $2$ with which we could compensate agent $i$ that allows the social welfare to compensate fully for this loss:  $(a_1,a_2) \succ (b_1,b_2)$.

The axiom makes sense in cases where the goods are, say, food and clean water.  It would make less sense if the goods were broccoli and cauliflower.

In the lattice theoretic notation, the requirement is that if changing agent $i$'s consumption from $x$ to $x \wedge y$ leads to a strict social decrease, then moving from $x$ to $y$ must lead to a strict social decrease.  This motivates the following general concept, which slightly weakens the principle.

\textbf{$\geq$-Complements}:  Let $i\in N$, and let $a,b,c\in A$.  Suppose that $a_ic \succ (a\wedge b)_i c$.  Then $a_ic \succ b_ic$.  %Then $(a\wedge b)_ic\succeq b_ic$.

%The axiom is motivated by the idea that the ``goods'' in question cannot substitute for each other.  Thus, in the $\mathbb{R}^2_+$ case, insufficient water for an agent could not be substituted for by giving that agent more housing.  Thus, we think of the goods as being complementary, at least in terms of sufficiency.

With this in mind, let us now discuss the concepts of interest here.  One was already mentioned before:  say that a binary relation $\succeq$ is \emph{threshold sufficientarian} if there exists $\beta \in A$ for which $x\succeq y$ iff $|\{i:x_i \geq \beta\}|\geq |\{i:y_i \geq \beta\}|$.

Our next definition allows us to capture a modest generalization of threshold sufficientarianism.  The generalization permits technical concepts such as 'strict sufficientarianism' described in the preceding section, as well as more economically relevant examples, such as commodity spaces where some commodities are deemed irrelevant.  

Say that $\mathcal{F}$ is a \emph{$\geq$-filter} if the following two properties are satisfied.
\begin{enumerate}
\item If $x\in \mathcal{F}$ and $y \geq x$, then $y \in\mathcal{F}$
\item If $x,y\in \mathcal{F}$, then $(x\wedge y) \in \mathcal{F}$.
\end{enumerate}

We call a sufficientarian binary relation whose sufficient set is a filter a \emph{generalized threshold sufficientarian} binary relation.  We now discuss the concept.

For a threshold sufficientarian binary relation, the set $S$ takes the form $S=\{a\in A:a \geq \beta\}$.  It is easily verified that every such set is a filter.  Nevertheless, there are filters which are not threshold sufficientarian.  Below, we recall the two basic examples we have in mind.

\begin{enumerate}
\item Let $A = \mathbb{Z}^2$.  Then the set $S=\{(a_1,a_2):a_1 \geq 0\}$ is a filter.
\item Let $A = \mathbb{R}$.  Then the set $S = \{a:a>0\}$ is a filter.
\end{enumerate}

The first of these is a multi-commodity environment, where one of the commodities (the second) does not matter; it is redundant in terms of sufficiency.  We would like to allow for this generalization in case some of the commodities, whose consumption can come in negative units, are not deemed necessary for the agents.  This type of generalized threshold sufficientarian rule is an economically meaningful generalization of the concept of threshold sufficientarianism.

The second of these is close to being sufficientarian, but the sufficient set does not contain its lower bound.  We referred to such an object as a ``strict'' sufficientarian binary relation earlier.  This example constitutes a largely technical generalization of threshold sufficientarianism.

Generalized threshold sufficientarian binary relations can hold both aspects of these two examples across different commodities, but essentially nothing else.

Our goal is to characterize generalized threshold sufficientarian binary relations through the use of $\geq$-complements.  As a corollary, we will also provide conditions under which we can say a binary relation is threshold sufficientarian.  The following definitions are useful.

A filter corresponding to a threshold sufficientarian rule is called a \emph{principal filter} in the mathematics literature.  

First, a simple definition.  We say that $(A,\geq)$ satisfies the \emph{descending chain condition (DCC)} if for every sequence $a_1 \geq a_2 \geq a_3 \ldots$, there exists $k$ such that for all $l,m\geq k$, $a_l=a_m$.

The DCC is a weaker condition than finiteness.  For example, $A=\mathbb{Z}_+^n$ with its usual pointwise order satisfies the DCC.  The relevance of the condition is that it is necessary and sufficient for all filters to be principal filters.

Nevertheless, we may choose to work with semilattices that do not satisfy the DCC.  We can recover threshold sufficientarianism by working with a slightly stronger notion of a semilattice.  We say that $(A,\geq)$ is \emph{meet-complete} if every nonempty subset $B \subseteq A$ has a greatest lower bound, which we denote by $\bigwedge B$.  For example, $\mathbb{R}_+^n$ is meet-complete, whereas $\mathbb{R}^n$ is not.

The following axiom is meaningful for meet-complete semilattices only.

\textbf{Descending continuity}:  Let $(\Lambda,\geq')$ be a nonempty linearly ordered set, and let $\{x^{\lambda}\}_{\lambda\in\Lambda} \in A^N$ be a set of allocations for which for all $\lambda,\lambda' \in \Lambda$ and all $i\in N$, if $\lambda' \geq' \lambda$ then $x^{\lambda}_i \geq x^{\lambda'}_i$ and $x^{\lambda} \sim x^{\lambda'}$.  Then $(\bigwedge_\Lambda x^\lambda_1, \ldots, \bigwedge_\Lambda x^\lambda_n)\sim x^{\lambda'}$ for any $\lambda' \in \Lambda$.

The following is our main characterization of threshold sufficientarianism.

\begin{theorem}A binary relation $\succeq$ is generalized threshold sufficientarian iff it satisfies weak order, symmetry, separability, sufficientarian judgment, monotonicity, and $\geq$-complements.  \end{theorem}

\begin{proof}We know by Theorem~\ref{thm:mainresult} that $\succeq$ is monotone-sufficientarian, let us call its corresponding sufficient set $S$.  Our goal is to establish that $S$ is a filter.  That it is closed upwards follows as $S$ is a monotone set.  Now, suppose that $a,b\in S$.  We want to claim that $(a \wedge b) \in S$, so suppose by means of contradiction that it is not.  

Then for any $c\in A$, we have $a_i c \succ (a\wedge b)_i c$, by the definition of sufficientarianism.  Therefore, by $\geq$-complements, $a_i c \succ b_i c$, which means that $b\notin S$, a contradiction.  So $S$ is a filter.

Finally, it is obvious that if $\succeq$ is generalized threshold sufficientarian, then it satisfies $\geq$-complements.  \end{proof}

The following proposition qualifies when we additionally can assert that the binary relation is threshold sufficientarian.

\begin{proposition}\label{prop:dcc}The following statements are true.
\begin{enumerate}
\item If $(A,\geq)$ satisfies the DCC, then a binary relation is threshold sufficientarian iff it is generalized threshold sufficientarian.
\item If $(A,\geq)$ is a meet-complete semilattice, then a binary relation is threshold sufficientarian iff it is generalized threshold sufficientarian and it satisfies descending continuity.
\end{enumerate}
\end{proposition}

\begin{proof}Both of these arguments are standard, for the first, see for example \citet{davey_priestley_2002}.  Let us prove the first claim.  Suppose that $(A,\geq)$ has the DCC, and suppose that $\succeq$ is generalized threshold sufficientarian, with sufficient set $S$ that is a filter.  We claim that $S$ is a principal filter.  Suppose by means of contradiction that it is not.  Then there is no $\beta\in S$ for which $S = \{a\in A:a \geq \beta\}$.  Therefore, for any $a\in S$, there exists $b\in S$ for which $b \geq a$ is false.  As a consequence, we obtain that $a > (a\wedge b)$, and since $S$ is a filter, $(a\wedge b)\in S$.  Consequently, for every $a\in S$, there is $a'\in S$ for which $a > a'$.  Now let $a_1 \in S$ be arbitrary, and inductively construct $a_n\in S$ from $a_{n-1}\in S$ so that $a_{n-1} > a_n$.  This sequence establishes that $(A,\geq)$ violates the DCC.

Conversely, suppose that the DCC is violated for $(A,\geq)$, and let $a_1>a_2>\ldots$ be a strictly decreasing sequence.  Define $S=\bigcup_n \{a\in A:a\geq a_n\}$.  Observe $S$ is a filter but cannot be principal:  if $S$ were principal, there would exist $\beta\in S$ for which for all $a\in A$, $a\geq \beta$.  But since $\beta \in S$, by definition of $S$, there exists $n$ for which $\beta\in \{a\in A:a \geq a_n\}$.  Hence $\beta\geq a_n > a_{n+1}$, a contradiction.

%For the second part, it is obvious that every threshold sufficientarian binary relation satisfies descending continuity.  For the converse, define $\beta \equiv \bigwedge S$.  We claim that $\beta \in S$.  Suppose not and well-order $S$ according to some relation $\lhd$, where $b'\lhd b$ means $b'$ is smaller than $b$.  If $b,b'\in S$ where $b'\lhd b$, then $\bigwedge\{a\in S:b'\lhd a\}\leq \bigwedge\{a\in S:b\lhd a\}$.  Let $b^*\in S$ be the $\lhd$-smallest element $b\in S$ for which $\bigwedge\{a\in S:b \lhd a\}\notin S$.  By descending continuity, there must be such an element; otherwise, $\bigwedge\{a\in S:b\lhd a\}$ is itself decsending with respect to $\lhd $ in $b$, each element of which is a member of $S$, yet $\beta=\bigwedge_{b\in S}\bigwedge\{a\in S:b\lhd a\}\notin S$, in contradiction to descending continuity.

%We have that $b^* \in S$, and further that $\bigwedge_{\{b:b^*\lhd b,b\neq b^*\}}\{a\in S:b\lhd a\}\in S$, where the latter follows by descending continuity.  These two hypotheses together with the fact that $S$ is a filter establish that $\bigwedge\{a\in S:b^*\lhd a\}\in S$, a contradiction. 

For the second part, it is clear that any threshold sufficientarian ranking obeys descending continuity.  Next, we may without loss assume that $S\neq \varnothing$ (the case of $S = \varnothing$ delivers the same representation as the case of $S=A$).  Observe that descending continuity implies that any $\geq$-chain $Y\subseteq S$ is bounded below by $\bigwedge Y\in S$.  Consequently, Zorn's Lemma implies that there is a $\geq$-minimal element $\beta\in S$.  We claim that for all $a\in S$, $a \geq \beta$, for if not, then there is $a\in S$ for which $\beta > a \wedge \beta $, contradicting minimality of $\beta$ as $a \wedge \beta \in S$ because $S$ is a filter.%it is obvious that every threshold sufficientarian binary relation satisfies descending continuity.  For the converse, define $\beta \equiv \bigwedge S$.  We claim that $\beta \in S$.  Suppose not and well-order $S$ according to some relation $\lhd$, where $b'\lhd b$ means $b'$ is smaller than $b$.  If $b,b'\in S$ where $b'\lhd b$, then $\bigwedge\{a\in S:a\lhd b\}\leq \bigwedge\{a\in S:a\lhd b'\}$.  Let $b^*\in S$ be the $\lhd$-smallest element $b\in S$ for which $\bigwedge\{a\in S:a \lhd b\}\notin S$.  By descending continuity, there must be such an element; otherwise, $\bigwedge\{a\in S:a\lhd b\}$ is itself decsending with respect to $\lhd $ in $b$, each element of which is a member of $S$, yet $\beta=\bigwedge_{b\in S}\bigwedge\{a\in S:a\lhd b\}\notin S$, in contradiction to descending continuity.

%We have that $b^* \in S$, and further that $\bigwedge_{\{b:b\lhd b^*,b\neq b^*\}}\{a\in S:a\lhd b\}\in S$, where the latter follows by descending continuity.  These two hypotheses together with the fact that $S$ is a filter establish that $\bigwedge\{a\in S:a\lhd b^*\}\in S$, a contradiction.  
\end{proof}

%\end{proof}

\begin{remark}In terms of descending continuity, it is often the case that it is enough to state a weaker version which applies to sequences only.  That is, if $\{x^m\}$ is a decreasing sequence of allocations, for which $x^m \sim x^{m'}$ for all $m,m'\in\mathbb{N}$, then $\bigwedge_m x^m\sim x^{m'}$ for any $m'$.  This axiom would be sufficient, for example, for $\mathbb{R}_+^n$ with its usual order.  Mathematically, the restriction to arbitrary linearly ordered sets (we could replace the axiom by nets) is necessary, however.  For example, if $A$ is the set of all subsets of some uncountably infinite set, then a binary relation defined by a sufficient set $S$ consisting of all sets whose complements are at most countable would satisfy the sequential version but not the general version used in Proposition~\ref{prop:dcc}.\end{remark}

\section{On endogenous leximin and a weakening of sufficientarian judgment}\label{sec:leximin}

This section is devoted to studying the relation between sufficientarianisn and the leximin ranking.  The leximin ranking $\geq_l$ is a ranking on $\mathbb{R}^n$, which asserts that for any $x,y\in \mathbb{R}^n$, $x \geq_l y$ if, when each of $x$ and $y$ are rearranged in nondecreasing order, $x$ dominates $y$ according to the classical lexicographic relation.  Thus, it ranks nondecreasing rearrangements lexicographically. 

Our setup is abstract and thus the leximin ranking is not defined on $A^N$ absent some ordering on $A$.  Nevertheless, there are obvious parallels between sufficientarianism and the leximin ranking.  We detail these below.

First, let us say that a weak order $\succeq$ on $A^N$ is an \emph{endogenous leximin} ranking if the following is true.  There exists a weak order $\geq^*$ over $A$ for which, for any $x,y\in A^N$ the following two assertions are true:
\begin{enumerate}
\item If $x_n \geq^* x_{n-1} \ldots \geq^* x_1$ and $y_n \geq^* y_{n-1} \ldots\geq^* y_1$, then $x \succeq y$ if either for all $i\in N$, $x_i =^* y_i$, or there exists an $i\in N$ such that for all $j<i$, $x_j =^* y_j$ and $x_i >^* y_i$.
\item For any permutation $\sigma:N\rightarrow N$, $x \sim x\circ \sigma$.
\end{enumerate}

For a given weak order $\geq^*$ on $A$, the two properties uniquely pin down a weak order over $\succeq$:  for any pair $x,y\in A^N$, it is enough to assume by point (2) that $x$ and $y$ are ``nondecreasing,'' and then rank them according to (1).  Transitivity of $\succeq$ ensures that the ranking of the nondecreasing arrangements coincides with the ranking of $x$ and $y$.

If, in the definition of endogenous leximin, the ranking $\geq^*$ has a utility representation, the connection with the classical leximin ranking becomes clear.  Suppose that $\succeq$ is an endogenous leximin ranking, and $\geq^*$ has a utility representation $u$, so that $a \geq^* b$ iff $u(a) \geq u(b)$.  Then $x \succeq y$ iff $(u(x_1), \ldots u(x_n)) \geq_l (u(y_1),\ldots,u(y_n))$.  The phrase ``endogenous'' refers to the fact that the weak order $\geq^*$ over $A$ could be anything whatsoever, and need not be representable via a numerical scale.

We are unaware of any abstract studies of the endogenous leximin concept, though the idea has obviously existed in the litearture for some time.  \citet{sprumont} studies a special case of it, where the ranking $\geq^*$ is tied to individual preferences via a kind of ``agreement'' property.  More generally, the monograph by \citet{fleurbaey} describes many applications of the idea, where instead of a ranking $\geq^*$ on $A$, a ranking $\geq^*$ on $A\times N$ is implicitly considered, and usually based on classical utility constructions.  By contrast, individual preferences do not appear in our work.

The following proposition is immediate, so we state it without proof.  

\begin{proposition}\label{prop:endo}A ranking $\succeq$ is sufficientarian if and only if it is endogenous leximin with weak order $\geq^*$, where $\geq^*$ has at most two equivalence classes.\end{proposition}

Proposition~\ref{prop:endo} by itself does not say much about sufficientarian rankings.  Any anonymous Bergson-Samuelson social welfare function applied to a ranking with two equivalence classes would give the same result.  This is because there is only one monotone and symmetric method of ranking vectors of two-values.

Nevertheless, here we show that an obvious weakening of sufficientarian judgment is closely related to endogenous leximin.  Instead of requiring that a loss in wealth for one agent cannot be compensated for partially or fully by a change in some other agent's consumption, we require instead only that it cannot be fully compensated.

\textbf{Weak sufficientarian judgment:} Let $a,b,c\in A$ and suppose that $i\neq j$.  Then $b \succ a_i b$ implies $b \succ c_j (a_i b)$.  

The following result establishes that endogenous leximin emerges upon weakening sufficientarian judgment to weak sufficientarian judgment.  It is not a full characterization, nevertheless, it demonstrates a connection of leximin to sufficientarianism.

\begin{proposition}Any endogenous leximin binary relation $\succeq$ satisfies weak order, symmetry, separability, and weak sufficientarian judgment.  If $|N|=2$, there are no other binary relations satisfying the four properties.\end{proposition}

\begin{proof}Verification of the four properties for any endogenous leximin binary relation is routine.  Separability is the only nontrivial property to verify, but this can easily be established using the fact that the leximin ranking is itself separable; see \emph{e.g.} \citet{daspremont,deschamps,barbera1988}.

For the converse when $|N|=2$, define $\geq^*$ in the obvious way, namely $a \geq^* b$ if for any $c$ and any $i\in N$, $a_ic \succeq b_ic$.  Separability ensures $\geq^*$ is well-defined while weak order ensures that $\geq^*$ is also a weak order.  

Now, let $x,y\in A^N$.  By symmetry and transitivity, suppose without loss that $x_1 \geq^* x_2$ and $y_1 \geq^* y_2$. 

We need to ensure that 
\begin{enumerate}
\item if $x_2 >^* y_2$, then $x \succ y$ 
\item if $x_2 \sim^* y_2$ and $x_1 >^* y_1$, then $x \succ y$.
\item if $x_2 \sim^* y_2$ and $x_1 \sim^* y_1$, then $x \sim y$.
\end{enumerate}

The third case is obvious by separability and transitivity.  In terms of the second case, we have that $x_1 >^* y_1$ and $x_2 \geq^* y_2$, so again by separability and weak order we must have $x \succ y$.

For the first case, if $x_1 \geq^* y_1$, then the argument is identical to the second case and we are done.  So finally suppose that $x_2 >^* y_2$ and $y_1 >^* x_1$.  Observe that $x =(x_1,x_2)\succeq (x_2,x_2)$, by separability and $x_1 \geq^* x_2$.  Again by separability, $(x_2,x_2) \succ (x_2,y_2)$.  Now by weak sufficientarian judgment, $(x_2,x_2) \succ (y_1,y_2)=y$.  Consequently by transitivity, we we establish that $x\succ y$.  \end{proof}

%Suppose that $x_1 \geq^* y_1$ and $x_2 \geq^* y_2$.  Then separability and transitivity immediately imply $x \succeq y$, consistent with endogenous leximin based on $\geq^*$.  Similarly, if either of the rankings is strict, say $x_1 >^* y_1$ and $x_2 \geq^* y_2$, then $ x\succ y$ follows from weak order and separability.

%Suppose now, and without loss again by symmetry and transitivity, that $x_1 >^* y_1$ and that $y_2 >^* x_2$, so that $x_1 >^* y_1\geq^* y_2 >^* x_2$.  We claim that $y \succ x$.  To see this, observe that $y =(y_1,y_2)\succeq (y_2,y_2)$, by separability.  Again by separability, $(y_2,y_2) \succ (y_2,x_2)$.  Now by weak sufficientarian judgment, $(y_2,y_2) \succ (x_1,x_2)=x$.  Consequently by transitivity, we we establish that $y\succ x$.  \end{proof}

Finally, we show that the four axioms are not enough to characterize endogenous leximin when $|A|\geq 3$.  We conjecture that straightforward strengthenings of weak sufficientarian judgment or of the other axioms should suffice for such a characterization, but we leave this to future research.

\begin{example}
Let $N=\{1,2,3\}$, and $A = \{a,b,c\}$.
Define a binary relation $\succeq$ on $A^N$ by
$(c,c,c) \succ (b,c,c) \succ (b,b,c) \succ (a,c,c) \succ (b,b,b) \succ (a,b,c) \succ (a,b,b) \succ (a,a,b) \succ (a,a,a)$, where symmetry and transitivity characterize the remaining parts of the binary relation.  It is straightforward to verify that all of the axioms are satisfied here.

It is also easy to see that, were $\succeq$ endogenous leximin, then $c >^* b>^* a$ would be necessary, as $(c,c,c) \succ (b,c,c)$ and $(b,b,b) \succ (a,b,b)$.  Yet endogenous leximin would then require that $(b,b,b) \succ (a,c,c)$, which is false.  \end{example}

\section{Conclusion}\label{sec:conclusion}

In this work, we have studied the concept of sufficientarianism applied to an abstract framework and also to more structured economic environments.  Sufficientarianism was shown to be the unique binary relation satisfying a collection of axioms, namely, symmetry, separability, and sufficientarian judgment.

%Sufficientarianism is to be understood as a ``pre-ranking,'' imposed to rule out profiles that should never be chosen in the presence of other feasible profiles.

There are several possible extensions for our work.  As a first point, our work is fixed population, and might amount to 'head-counting.'  A work that allows one to consider multidimensional allocations across different populations would be similarly interesting; see, for example, \citet{Bossert2021, Bossert2022}.  

Second, sufficientarianism lends itself to a natural cardinal ranking, in which we only count the number of people who reach a given threshold.  One could consider imposing a random threshold and studying the expected number of people to reach the threshold.  In the case of $X=\mathbb{R}_+$, we could clearly end up with the set of all (linear) rankings satisfying a weak form of monotonicity in doing so.  But in higher dimensional environments, we simply have not studied what this approach would generate.  It is conceivable that the sufficientarian rankings generate as their expectation some type of supermodular and increasing rankings.  We leave this to future research.

Threshold sufficientarianism is meaningful for any meet-semilattice, though we have only provided a characterization in semilattices that are meet-complete. Characterizing it more generally may be interesting. Of course, weakening axioms, especially that of completeness of the social ranking, might prove interesting as well.

%Finally, we have motivated our concept of sufficientarianism as a kind of ``pre-ranking,'' to be imposed prior to other, more finely-grained rankings.  We do not know what happens when taking this literally, and imposing a Paretian ranking ``on top'' of a sufficientarian ranking.  For example, given a sufficientarian ranking $\succeq$ and a Paretian ranking $\succeq^*$, this approach suggests ruling out any $(N,x)$ for which either there exists a feasible $y$ for which $(N,y)\succ (N,x)$ or for which there exists a feasible $y$ for which $(N,y) \sim (N,x)$ and $(N,y)\succ^* (N,x)$.  Again, we leave this to future research.

%Summary
%Limitations
%Implications for Future Research

\newpage

\bibliography{Bibliography}
%\cite{Roger2003}
%\cite{Bossert2022}
%\cite{Mariotti_2021}
\end{document}